\documentclass[conference]{IEEEtran}
\ifCLASSOPTIONcompsoc
  \usepackage[nocompress]{cite}
\else
  \usepackage{cite}
\fi
\ifCLASSINFOpdf
   \usepackage[pdftex]{graphicx}
\DeclareGraphicsExtensions{.pdf, .png, .jpg}
\else
\fi

\pdfoutput=1
\usepackage{float}
\usepackage{url}

\usepackage{amsthm}
\newtheorem{definition}{Definition}
\newtheorem{theorem}{Theorem}
\newtheorem{lemma}[theorem]{Lemma}

\newtheorem{assumption}{Assumption}

\usepackage{mathtools}

\begin{document}
\title{YAC: BFT Consensus Algorithm for Blockchain}

%NOTE: authors are ommited for peer review
\author{Fedor Muratov, Andrei Lebedev, Nikolai Iushkevich,\\ Bulat Nasrulin, Makoto Takemiya\\
\footnotesize Soramitsu\\
\footnotesize % \affaddr{Shenzhen 518055, China. }\\
\{fyodor, andrei, nikolai, bulat, takemiya\}@soramitsu.co.jp \\}%
%E-mail: see http://www.michaelshell.org/contact.html
%\IEEEcompsocthanksitem J. Doe and J. Doe are with Anonymous University.}% <-this % stops a space
%}


% \markboth{2018 Crypto Valley Conference on Blockchain Technology (CVCBT)}%
% {2018 Crypto Valley Conference on Blockchain Technology (CVCBT)}

\IEEEcompsoctitleabstractindextext{%
\begin{abstract}
Consensus in decentralized systems that asynchronously receive events and which are subject to Byzantine faults is a common problem with many real-life applications. Advances in decentralized systems, such as distributed ledger (i.e., blockchain) technology, has only increased the importance of finding performant and secure solutions to consensus of state machine replication in decentralized systems.

YAC is a practical decentralized consensus algorithm, that solves the problems of inefficient message passing and strong leaders that occur in classical Byzantine fault tolerant consensus algorithms. The algorithm is open source and currently is used to provide Byzantine fault tolerant consensus for the Hyperledger Iroha blockchain project. We provide proofs of safety and liveness, as well as empirical results showing that our algorithm can scale to dozens of validating peers.

%We provide an example of consensus round with 4 peers in the network, where one peer is malicious, and prove that safety and liveness constraints are guaranteed.
\end{abstract}

}

\maketitle

\IEEEdisplaynotcompsoctitleabstractindextext
\IEEEpeerreviewmaketitle

\section{Introduction}
\label{sec:introduction}

Decentralized systems built around blockchain technology have received much attention in recent years. Satoshi Nakamoto \cite{nakamoto2008bitcoin} proposed the blockchain structure as a way for timestamping sets of transactions without a central authority, which he used to create a purely decentralized, digital cash. The advantage of Nakamoto's proposed method for decentralized consensus (hereafter, \textit{Nakamoto consensus}) is that network participants are able to operate in an untrusted or semi-trusted environment with a strong adversary assumption, that is, in the presence of Byzantine faults among network participants.

However, the proposed method of using Proof-of-Work \cite{DworkN92} is too resource intensive to be practical for many applications, and the reliance on using Proof-of-Work to determine the ``correct'' blockchain out of several competing blockchain states, leads to a lack of transaction finality that is caused by the probabilistic nature of Proof-of-Work. Therefore, researchers have given renewed attention to Byzantine fault tolerant algorithms that function in a deterministic way under the assumption of less than $f$ faulty participants in the network consensus, due to guarantees of transaction finality and computational efficiency.

While Byzantine fault tolerant consensus algorithms have a long pedigree spanning several decades, a majority of previous research has not focused on aspects peculiar to blockchain technology, such as creating practical systems that are censorship and politically resistant.

In this paper we present a novel, practical Byzantine fault tolerant consensus algorithm that has a modular architecture and simple implementation, called Yet Another Consensus. Our implementation is open source and included in the Hyperledger Iroha blockchain platform\footnote{\url{https://github.com/hyperledger/iroha}}. Empirical results show that our solution achieves both low latency for transactions as well a reasonably high transaction throughput.

\section{Related work}

Decentralized systems are distributed systems that are designed to have no central authority. Instead, the rules of the system create a decentralized platform where data can be agreed upon, through relying on some mechanism for distributed consensus. Blockchain platforms use decentralized consensus to maintain consistency across a distributed state machine, which can be used to perform completely decentralized payments \cite{nakamoto2008bitcoin} or even Turing complete computations \cite{wood2014ethereum}, realizing the creation of decentralized applications. Consensus is the \textit{sine qua non} of distributed systems \cite{CorreiaVNV11}, and has received large amounts of attention recently from the industry \cite{CachinV17}. It is the role of consensus in blockchain systems to guarantee that all non-faulty peers in a blockchain network perform the same state machine updates in the same order. More than 30 years of research in distributed systems has failed to create efficient, internet-scale consensus \cite{PassS17}.

Consensus \cite{lamport2005generalized} should guarantee \textit{liveness} of the system, \textit{security}, and \textit{convergence} (consistency) of data stored in the ledger. Liveness means that the system should never stop and should be able to recover from errors. Security means that non-faulty peers should not accept false data. Consistency means that all non-faulty peers should maintain or converge to the same global ordering and state. Various consensus algorithms have been proposed for different situations. Consensus algorithms are often discussed in terms of some weak synchrony assumption, and, as shown by Fisher et al. \cite{Fischer:1985ji}, no distributed consensus algorithm can give a deterministic solution in a fully asynchronous network.

Satoshi Nakamoto's proposal \cite{nakamoto2008bitcoin} for decentralized consensus uses the concept of an adaptable difficulty to create a completely decentralized computational problem that is very difficult to solve, yet trivial to verify the answer to. The difficulty of solving the Nakamoto Proofs-of-Work prevent Sybil attacks on decentralized systems, as the required computations are non-trivial. Other forms of Nakamoto consensus, such as Proof-of-Stake \cite{BentovPS16a} or Proof-of-Importance, use the concepts of globally known, random numbers (such as block hashes and public keys), along with finite values of non-zero cost, such as account balances in Proof-of-Stake or account importance scores in Proof-of-Importance, to determine what account has the privilege of creating a block at a given time.

%TODO: About BFT, failures in general.

The practical Byzantine fault tolerance (PBFT) \cite{Castro:1999jd} consensus algorithm, in the normal case, runs a three-phase protocol: pre-prepare, prepare, and commit. A client sends a request to one of the peers, who in turn broadcasts pre-prepare messages to the other peers. In the prepare stage, a prepare message is multicasted to all other nodes. When a replica receives $2f$ prepare messages, it matches with the pre-prepare message and multicasts a commit message. Once the replica receives $2f+1$ commit messages, that match the pre-prepare message, it changes the state to \textit{committed} and executes the message operation. Once the message is executed, a reply is sent to the client. The main purpose of a Byzantine fault tolerant consensus algorithm is to allow the system to be able to survive and continue work despite some of the machines exhibiting arbitrary faults. Although, PBFT is a consensus algorithm with proven security and liveness properties, it makes some assumptions that are not practical. The network overhead during consensus round does not allow scale the consensus protocol, limiting the throughput of the whole system. It was shown \cite{Miller2016} that PBFT can be attacked by an adversary using a simple scheduling mechanism, halting the consensus either completely, or forcing to wait long timeout when leader is partitioned and unsynchronized.

Zyzzyva \cite{KotlaADCW07} extended PBFT, avoiding the expensive three-phase commit protocol, utilizing fast track and  actively involving the client into the consensus process. In the best case scenario, the client sends the request to the leader, who is broadcasts it to all other replicas, and the client waits for replies from all the peers. If the client receives 3$f$+1 replies, it commits the transaction(s). In the case when the client receives between 2$f$+1 to 3$f$ messages, a regular consensus algorithm is used. The client is a major player, who is responsible for checking the integrity of replicas.

Hyperledger Fabric~\cite{fabric2018, fabric2017} implements a pluggable consensus on the order of transactions in the ordering service. At the time of this writing, Hyperledger Fabric includes the following consensus possibilities: a centralized server, crash fault tolerance, and Byzantine fault tolerance (using BFT-SMart \cite{Bessani:2014ei}). The main difference between BFT-SMart and PBFT is improved reliability and multicore computation for the evaluation of signatures. The ordering cluster uses $3f+1$ nodes, each of which collects current envelopes and runs BFT-Smart consensus. A transaction in Hyperledger Fabric consists of the following steps: (1) a client creates a transaction and sends it to all endorsing peers to execute the transaction against the current state, (2) endorsing peers transmit the result of the execution to the client, (3) the client collects all endorsements, assembles them into the transaction and broadcasts it to the ordering service, (4) the ordering service collects the transactions and relays the transaction proposal to the channel peers, (5) the channel peers validate the transactions and apply valid transactions to the current local state. When a transaction is validated by a peer, the result is sent to the client. A transaction is considered final only when the client receives enough transaction statuses from validating peers, which complicates the process and client.

BChain \cite{Duan:2014ce} is a chain-based (as opposed to broadcast-based) Byzantine fault tolerant consensus algorithm, which uses information about the topological order of validators in the network to determine the flow of messages sent to reach consensus. BChain prioritizes efficiency of message passing over latency of a transaction, because chain-based consensus algorithms validate transactions serially and not in parallel.

Tendermint~\cite{kwon2014tendermint} requires locking a certain amount of coins (stake), which cannot be spent during the duration of mining. Unlike other algorithms where validators are all equal, Tendermint validators are not required to be equal, with the voting power being proportional to the amount of coins locked. The consensus algorithm itself follows the PBFT three-phase commit message pattern, but uses a gossip protocol for blocks and transaction propagation. A validator commits a block only if it has the block and votes for it from two-thirds of the total voting power.

Algorand~\cite{Gilad:2017bp} is similar to Tendermint in that it has stake-weighted voting by valdiators, but it is unique in that it combines the stake values with a cryptographic sortition function to choose the voting committee. This makes the consensus algorithm more resistant to censorship than Tendermint. However, designing stake-weighted consensus algorithms is tricky~\cite{Kiayias:2017jb}, as this approach can open up new attack vectors.

Miller et al. ~\cite{Miller2016} argued that current weak synchronous consensus algorithms rely heavily on timing assumptions. They showed that an adversary can attack PBFT, causing it to either stop making any progress for consensus, or significantly slowing the protocol.

The HoneyBadger consensus algorithm is an asynchronous protocol, promising practical throughput.
The authors make the assumption that peers communicate with each other through authenticated trusted channels, and consensus is done only within a static community, without the possibility of extending it. Although, our consensus algorithm, YAC, relies on the timing assumption, we carefully evade problems of the PBFT consensus algorithm, by using a dynamic peer list.

% TODO: Final words, why creating YAC makes sense and why it is new. %

\section{Motivating Example}
\label{sec:example}

This section describes one round of YAC consensus execution with 4 peers (Alice, Bob, Clara, Deana), and follows the steps to reach consensus for a block of transactions.

To explain the YAC algorithm, we will follow an example of consensus along each each step of the pipeline. The pipeline starts with each client sending its own transactions to the ordering service (OS). It is the responsbility of the OS to collect all transactions, order them, and create a block proposal $P_1$. The ordering service then shares proposal $P_1$ with each peer in the network.

\subsection*{Alice's validation process}

Alice is a `validator" who validates a proposal $P_1$.  Alice tries to apply each transaction in the proposal to her local state. A transaction is considered valid if it is not ill-formed according to validation rules, and its application does not violate rules about global state (e.g., no account can have negative balances). The validation process is shown in Figure~\ref{fig:validation}. Alice creates a \textit{block} from all valid transactions and calculates hash of block $H_1$.

\begin{figure}[h!]
	\centering
    \includegraphics[width=0.9\columnwidth]{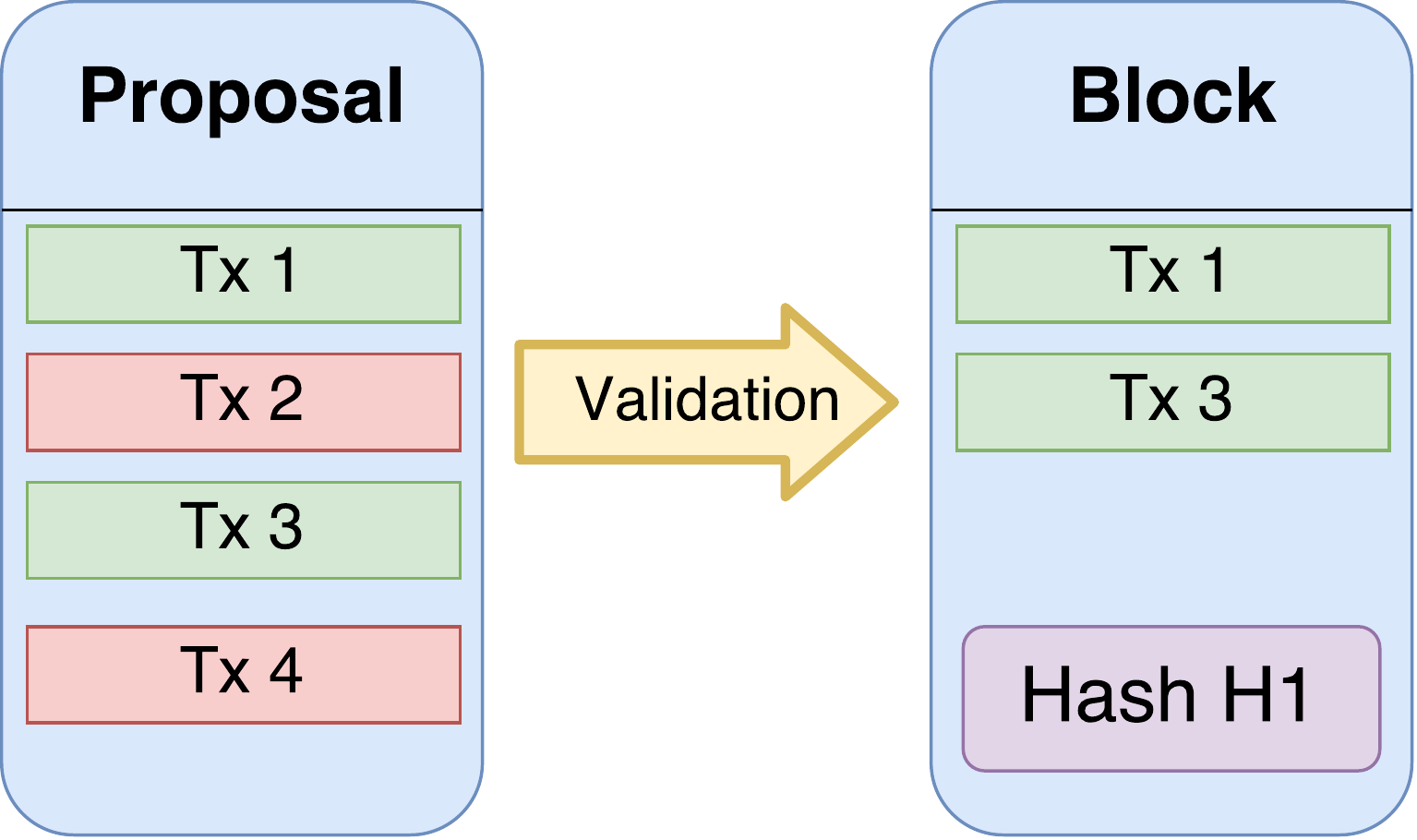}
    \caption{\textit{Validation process.} Alice removes invalid transactions.}
    \label{fig:validation}
\end{figure}

\subsection*{Order function and Alice's vote}

Alice now knows the proposal hash $H_1$ and the initial order of peers---$[A, B, C, D]$. Using hash $H_1$ as an input for the order function, Alice calculates the permutation for the current round, with the result $[C, D, A, B]$. The first peer in this round's permutation list is Clara. Alice creates a vote and propagates the vote to Clara. After sharing the vote, Alice switches her local state to \textit{waiting} for a commit message, up until some time delay (Figure~\ref{fig:state}).

\begin{figure}[h!]
	\centering
    \includegraphics[width=0.8\columnwidth]{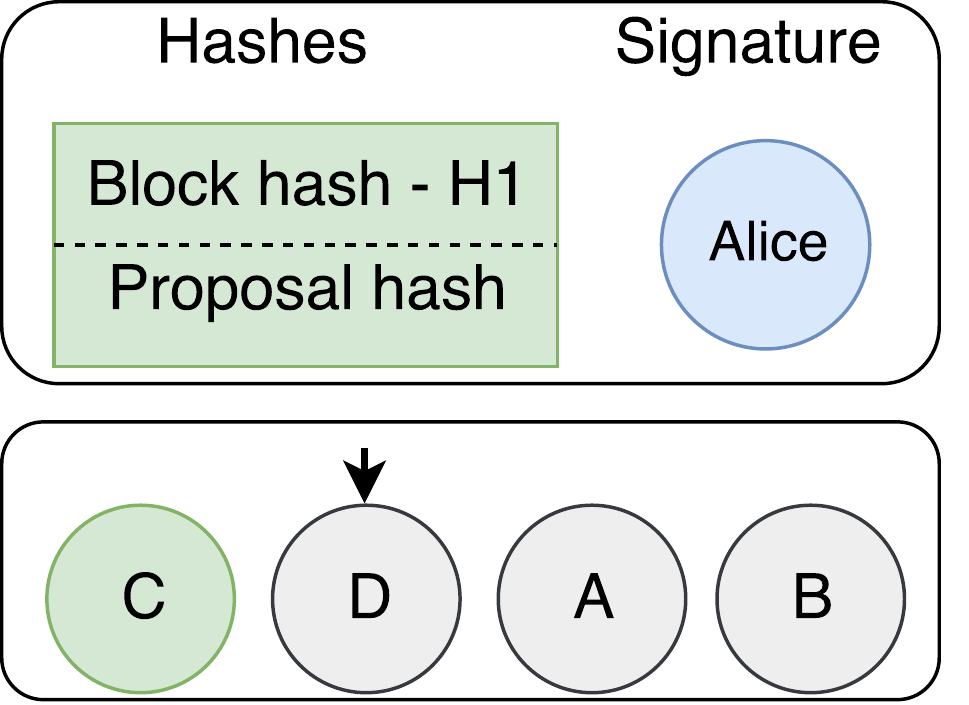}
    \caption{\textit{Alice's state.} Alice sends the generated vote to Clara, the first peer. The vote will be sent to Deana if no commit message is received after some time delay.}
    \label{fig:state}
\end{figure}

\subsection*{State of Clara}

Clara received Alice's vote. Assume that Clara calculated the same hash $H_1$ from the validation process of proposal $P_1$. Clara propagates the vote to herself, since Clara calculates the same order using the ordering function. Clara now has two votes: Alice's and Clara's. But there are no votes yet from Bob and Deana.

\subsection*{Deana shares the vote}

\begin{figure}[h!]
	\centering
    \includegraphics[width=0.9\columnwidth]{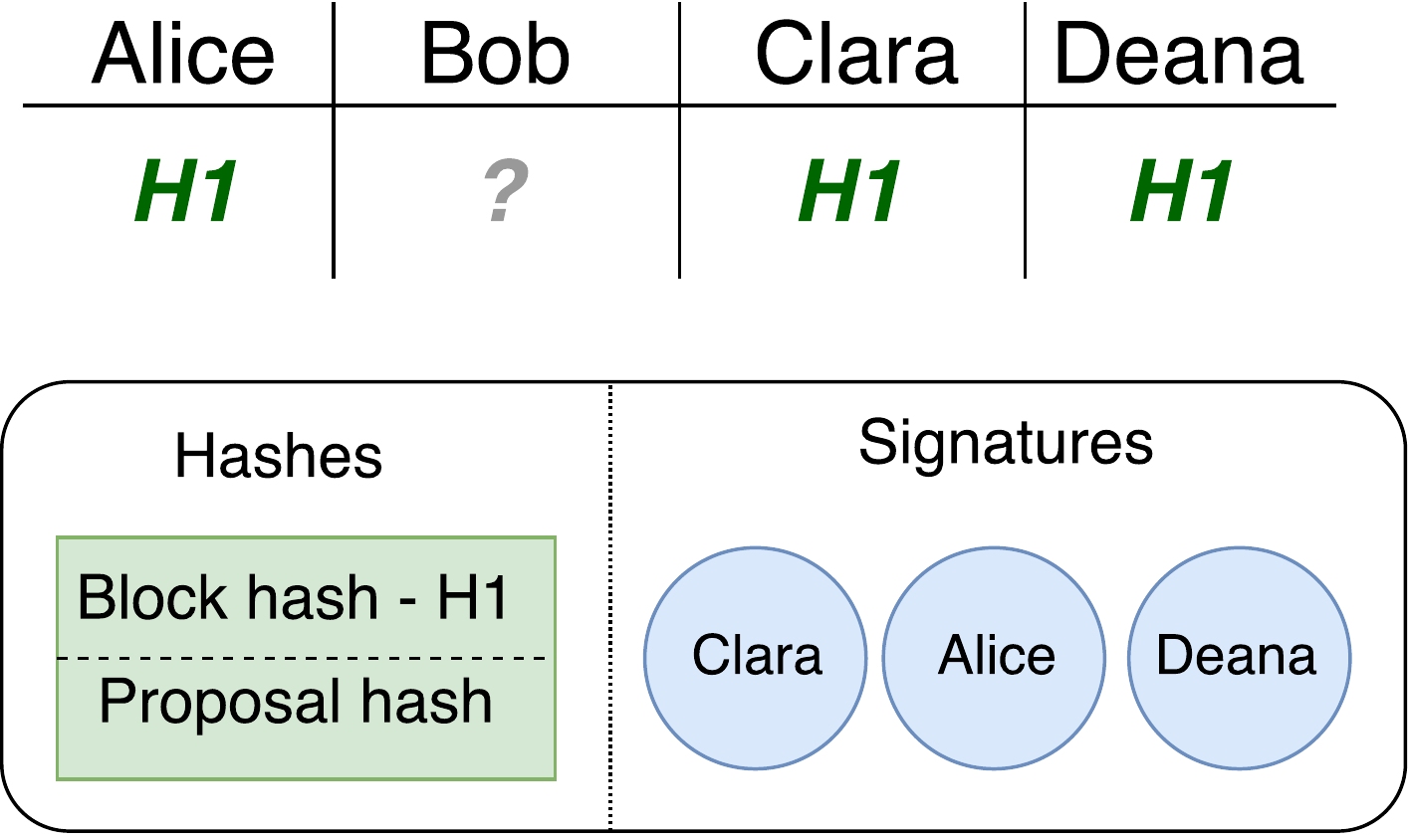}
    \caption{\textit{Clara's state.} Clara makes a commit message after receiving \textit{3} votes.}
    \label{fig:clara}
\end{figure}

Deana also computed hash $H_1$ and sent her vote to Clara. Now Clara has the required supermajority ($> 2/3$) of votes (Figure~\ref{fig:clara}). Clara broadcasts the commit message, with all the received votes from the peers (Figure~\ref{fig:broadcast}). Each peer, except Bob, receives a commit message and adds signatures to the block with the hash $H_1$ and updates the local state.

\begin{figure}[h!]
	\centering
    \includegraphics[width=0.9\columnwidth]{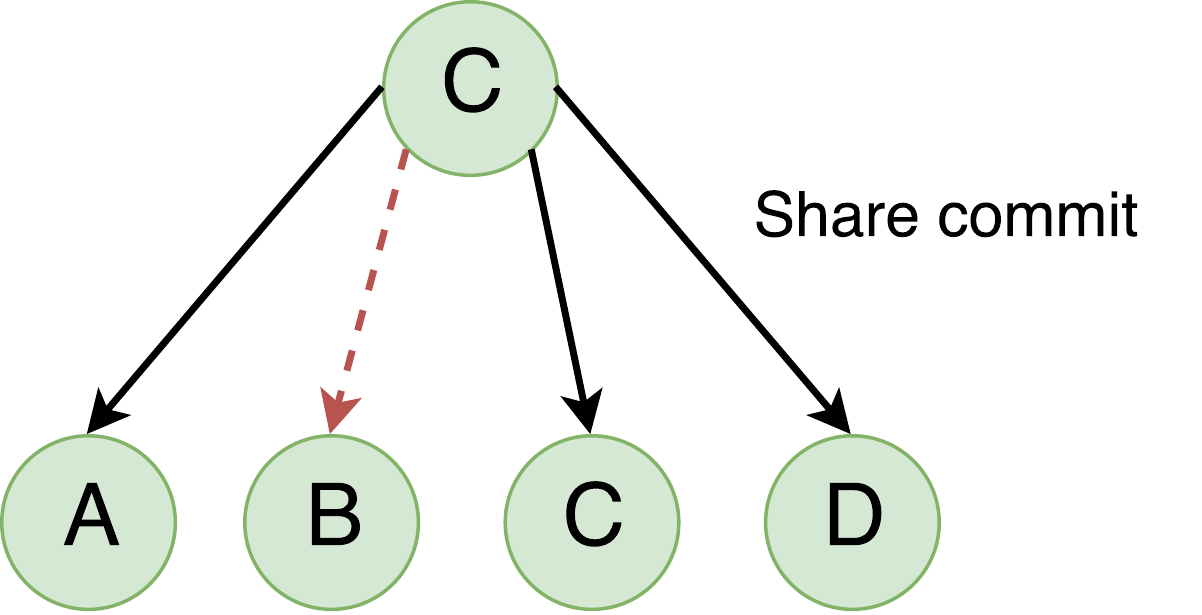}
    \caption{\textit{Commit broadcast.} Clara shares the commit message with all peers.
    Bob has connectivity issues and cannot receive the commit.}
    \label{fig:broadcast}
\end{figure}

\subsection*{What about Bob?}
\label{subsec:whatAboutBob?}

\begin{figure}[h!]
    \centering
    \includegraphics[width=0.8\columnwidth]{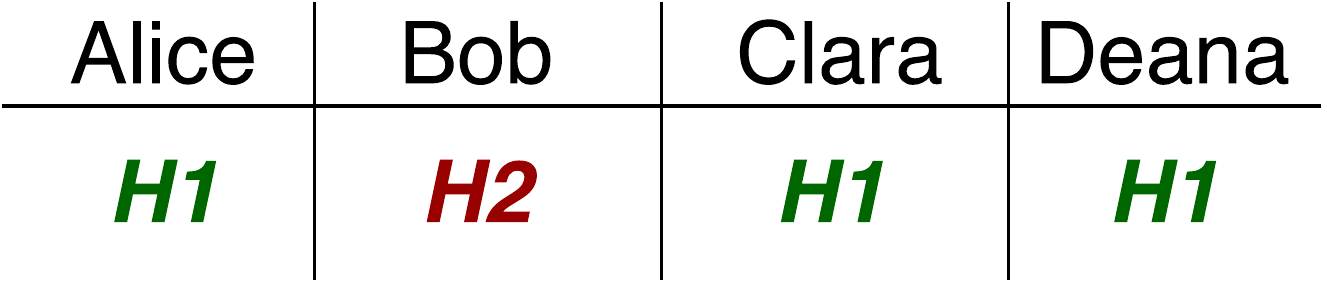}
    \caption{Alice's state after commit.}
    \label{fig:alice}
\end{figure}

Let's assume that Bob has an unstable connection to the network and he missed the previous and current rounds, including Clara's commit message. Bob did not propagate a vote with hash $H_1$ to Clara because he had an inconsistent state, or is malicious. Bob computes hash $H_2$, and gets a different peer order $[A, B, D, C]$ based on the hash $H_2$. Bob propagates his vote to Alice, and after Alice receives the vote, and since Alice has already received the commit message from Clara, she forwards Clara's commit message to Bob \textit{directly}. Bob verifies the commit message from Alice and applies it. He will now have the same state as others after the consensus round.

\section{YAC}

\subsection{Preliminaries}

Below we give an overview of the typical participants in YAC\@.

\subsubsection*{Client}
Each client is associated with one user who has a public key registered in the blockchain system. In general, the role of the client is to generate transactions and send them to the ordering service. Client is also deploying smart contracts, making queries to peers. 

\subsubsection*{Peer}

A peer is the network participant responsible for validating and reaching agreement on transactions in proposals and storing the negotiated transactions into blocks. Peers maintain the complete transaction history in order to validate proposals.

\subsubsection*{Ordering Service}

The ordering service is the functional module responsible for taking sets of transactions and creating block proposals. A block proposal contains a list of transactions for further agreement at peers. A proposal is a set of transactions to be be validated and voted on by peers. Ordering service is an abstract entity having multiple implementations. We discuss more on it in section~\ref{subsec:ordering}.

%TODO: \textbf{Add more details what are the requirements to the consensus in Iroha}

%TODO: \textbf{Add more details about features of YAC. }

%YAC works with a known set of validators.

%YAC algorithm was inspired by classical PBFT algorithm, significantly improving it.

%Classical leader-based algorithms have one obvious vulnerability: the leader is exposed to denial of service (DoS) attacks and can censor transactions or votes.

%The algorithm does not require the leader election and every peer can collect collaboration messages.

%Adversary cannot attack the network without the knowledge what the block hash is.

%We further add statistics such as response time, probability of maliciousness to improve the round time convergence in practice.

We provide the below definitions and general assumptions to use throughout this paper.

\begin{definition}
\label{definition:round}
    Round $r$ is a period of processing the $r$-th proposal by a peer. The round starts at invoking processing a proposal and finishes on committing a block. A round consists of 2 logical phases: validation of transactions in the proposal and agreement on the block derived from the proposal.
\end{definition}

\begin{definition}
\label{definition:honestPeer}
    Honest peer is a peer that tries to sync with the network, creates valid votes, commits and, never creates forks for any round.
\end{definition}

\begin{assumption}
\label{assumption:os}
    We assume that Ordering Service is Byzantine fault tolerant by itself and agreement on the order of transactions is honest. Also, all valid transactions from clients will eventually appear in some proposal. The Ordering Service guarantees that a proposal will be delivered for all peers in the network.
\end{assumption}

\begin{assumption}
\label{assumption:asycEnvironment}
    We use the assumption of an asynchronous environment~\cite{Attiya1990}; each message that is successfully sent will eventually received by the recipient.
\end{assumption}

\begin{assumption}
\label{assumption:crypto}
    Hashes and digital signatures used have cryptographic resistance and messages can't be altered by adversaries.
\end{assumption}

\begin{assumption}
\label{assumption:agreementOnSignatures}
    YAC guarantees BFT replication for the data---transactions in the block. However, different honest peers may apply different subsets of signatures for a block. This is a valid situation because each honest peer has proof of supermajority voting for a committed block.
\end{assumption}

\subsection{Overview}

For simplicity of the general pipeline explanation, we make several assumptions. First, we assume the client is known by the peers, and the client has a list of known peers to interact with. Second, client has its own keys stored securely on a device. Third, the client has permissions to execute a certain subset of commands/smart contracts on the blockchain (e.g., to transact). The general network pipeline of one round in a blockchain system with YAC can be described with the following steps:
\begin{enumerate}
\item[1.] A client forms a transaction with commands and signs it with their private key.
\item[2.] The client sends the transaction to a peer. The peer receives the transaction, performs stateless validation (i.e., verifies that it is not malformed), and relays it to the OS.
\item[3.] The OS generates a proposal and sends it to the peers. The proposal contains an ordered list of transactions that will be potentially added to the blockchain in this round.
\item[4.] The proposal is sent to the voting peers. Peers enter the collaboration phase, during which they exchange votes across the network and decide on a block. More details on the collaboration phase are provided in Section~\ref{sec:collab}.
\item[5.] The peer commits the block to their local block store.
\end{enumerate}

\subsection{Ordering phase}\label{subsec:ordering}
%
%TODO:  10/03/2018 move  separate to overview section (os pipeline) and preliminaries (assumptions, possible implementations)%
The OS collects transactions to include into a new proposal. A proposal is generated after the OS has collected some number of transactions, or after a time limit. The OS then broadcasts the proposal to all peers after the proposal is created.

The ordering service is an abstract concept that can have many possible implementations. The OS can be considered as one entity that is defined upon network creation. The ordering phase can be expressed using three steps: 1) clients share transactions with the ordering service, 2) the OS makes a proposal by combining received transactions, 3) the OS then shares the proposal with all validating peers.

%For our work, we use a censorship-resistant algorithm with guaranteed liveness, where one peer is elected a leader in each consensus round and is responsible for creating a block proposal.

\begin{figure}[H]
    \centering
    \includegraphics[width=0.9\columnwidth]{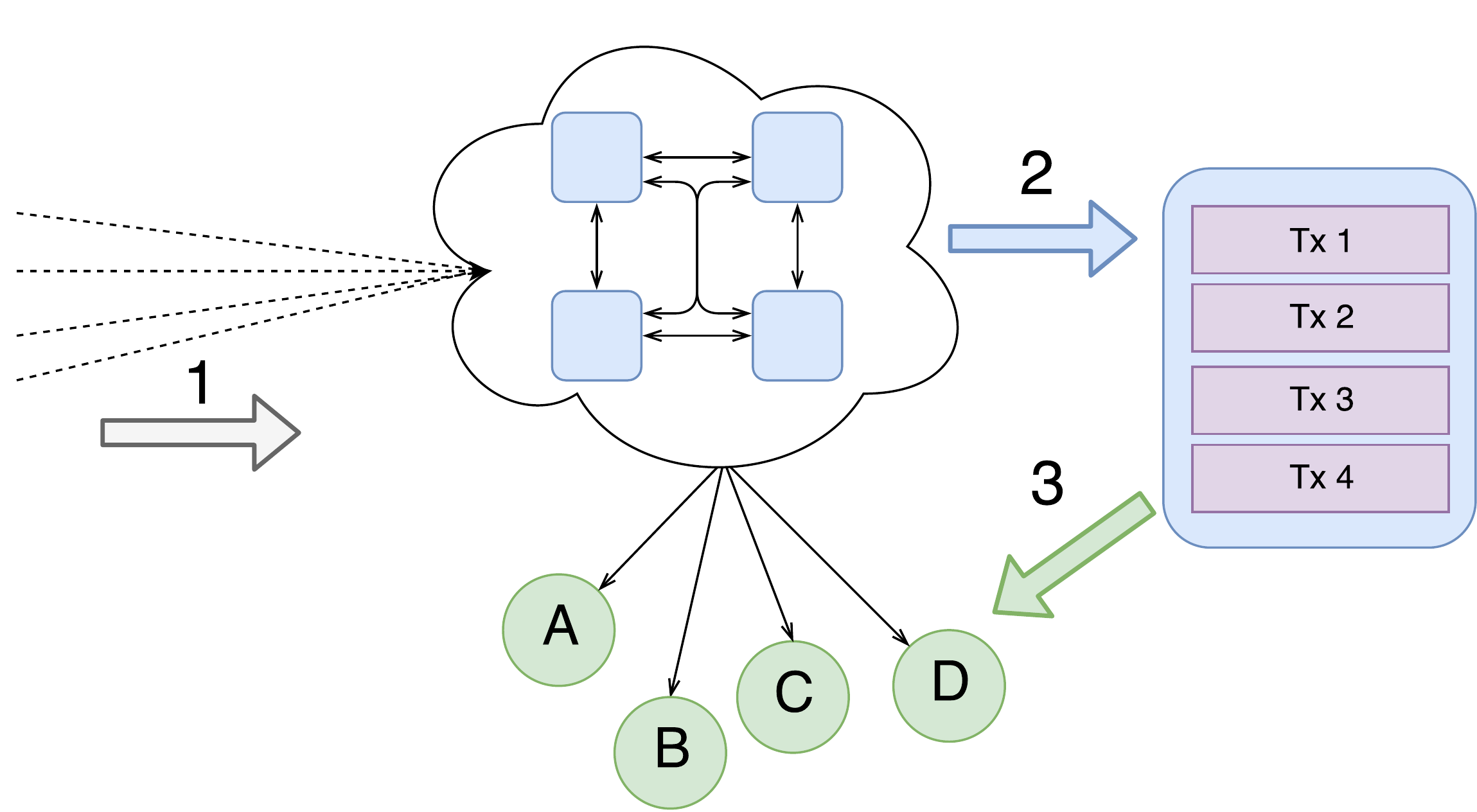}
    \caption{Ordering phase: 1. Clients send transactions to a peer with the ordering service; 2. the ordering service makes a proposal; 3. the ordering service shares the proposal with all peers.}
\end{figure}

\section{Collaboration phase}
\label{sec:collab}

%\textit{
%Section describes collaboration phase.
%Sunny-day case.
%Need to highlight separately that vote step not affected on peer that shares commit message.
%Failure cases: broken leader, reject case, strategies when shutdown greater then f.}

A peer calculates a \textit{verified proposal} after it receives a proposal from the ordering service. A verified proposal is a subset of transactions from the proposal, defined to be valid by the peer. The block, that is generated by a peer, consists of a proposal hash, transactions from the verified proposal, and additional metadata required for cryptographic validation of the chain.

The \textit{proposal hash} defines a unique proposal for each round of collaboration. The \textit{block hash} represents an intention of the peer to store a subset of transactions in the ledger. These hashes are required because different peers can calculate different blocks from the same proposal. A \textit{vote} message contains the pair of hashes, and a signature, which is used to authenticate a peer when a message is received from the network.

\subsection{Permutation function}
\label{subsec:permutationFunction}

When a peer votes for a block hash, it generates an order over the validating peers for the current round. The order is a permutation of peers that is required to propagate the vote in the network.

The order is generated by a function that takes the block hash and an initial list of peers as parameters. The order function is required to be pure and return uniformly distributed lists.

%TODO:  10/03/2018 extend the section %
%TODO:  10/03/2018 add pseudocode %

\subsection{Vote step}

%TODO:  10/03/2018 add pseudocode %

Vote messages for a proposal are sent to each peer in the specified order. A delay exists between each propagation. The process repeats until a valid commit or reject message is received from the network. One iteration of this process is called a \textit{vote step}. The process of propagation starts from the first peer after propagation to last peer in the order of validating peers.
%Commit and reject messages are defined by supermajority of peers.
%Supermajority is defined as more than two thirds of peers in current round.
%It is said that supermajority is achieved on a hash when a peer has collected a supermajority of votes for a hash.

\begin{figure}[H]
    \centering
    \includegraphics[width=0.7\columnwidth]{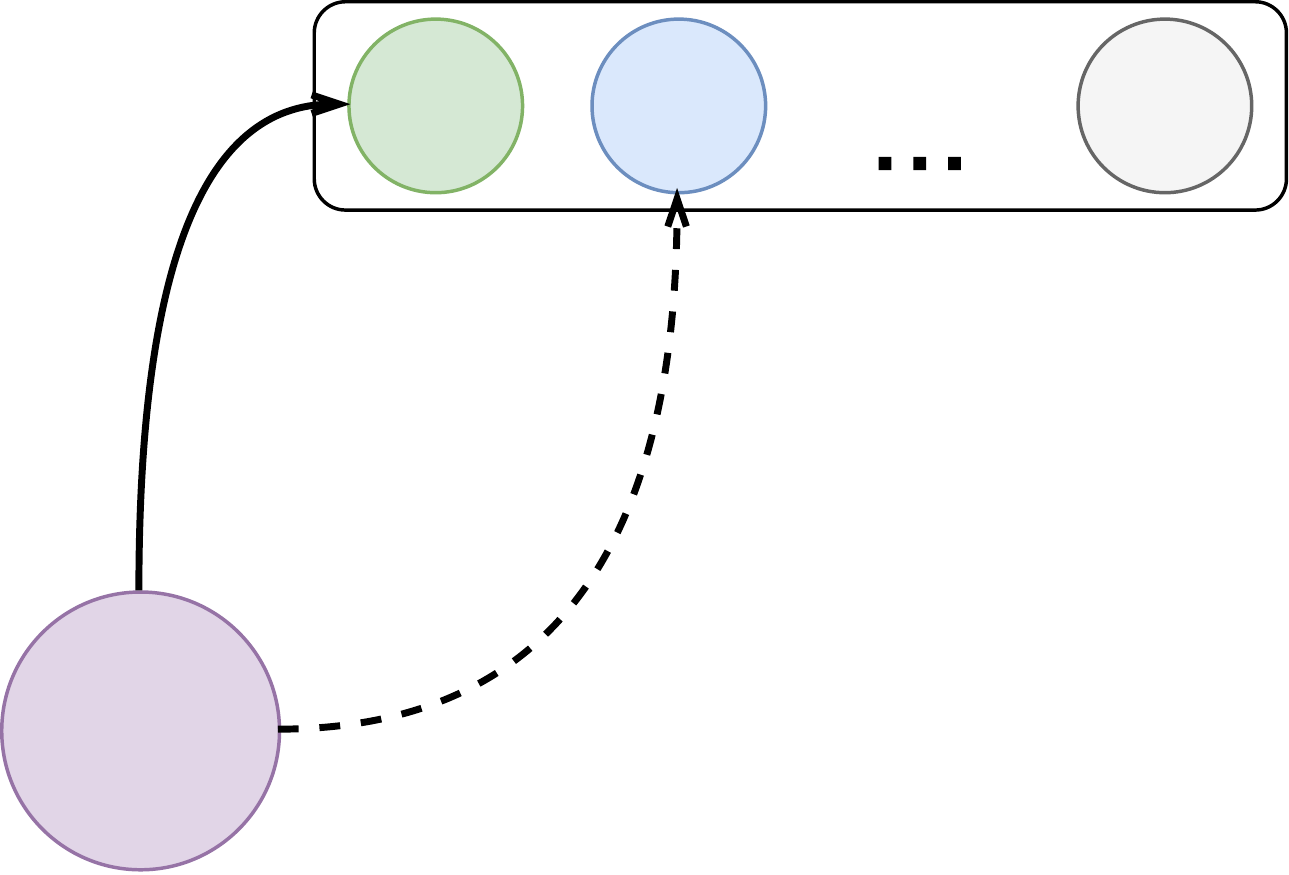}
    \caption{\textit{Vote step.} The vote message is sent to the first peer in the list (solid line), and then is sent to the next peer (dashed line), because of timer expiration in the vote step.}
\end{figure}

\subsection{Commit}
\label{subsec:commit}

\textit{Supermajority} is defined to be a number greater than two thirds of all peers in the network. A \textit{commit} message is a set of votes for one block hash, signed by a supermajority of peers. When some peer has collected a supermajority of votes for one hash, it broadcasts the commit message.

\subsection{Reject}

A \textit{reject} message is a set of votes that proves that peers will not collect a supermajority of votes for any block hash. This can be defined as if the sum of missing votes and votes for the most frequent block hash is less than the supermajority of the nodes. A peer broadcasts the reject message in the same way as the commit message.

\begin{figure}[h]
    \centering
    \includegraphics[width=0.7\columnwidth]{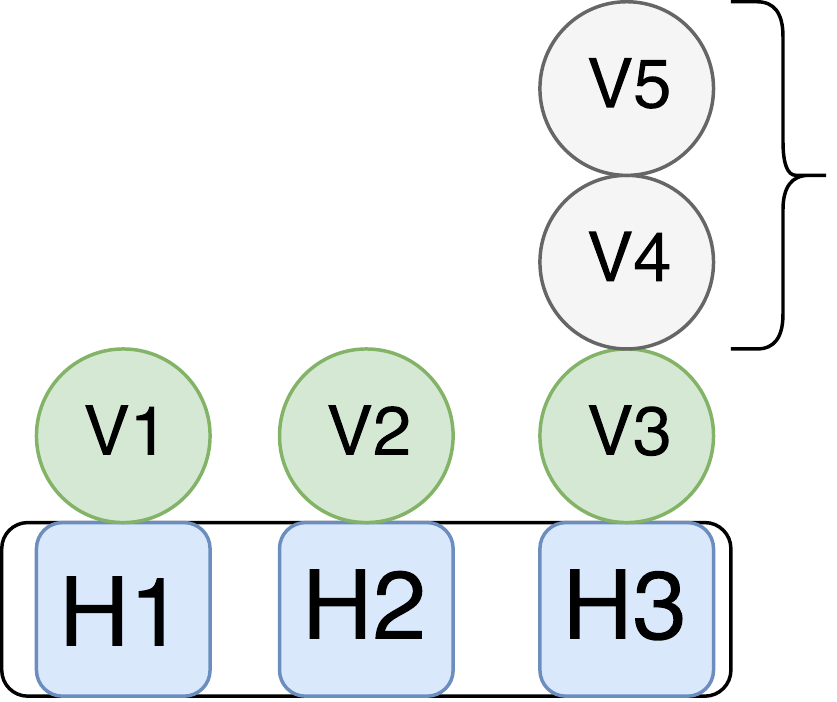}
    \caption{\textit{Reject.} Green votes have been collected by a validating peer, gray votes are missing from validating peers.}
\end{figure}

It is possible that a peer can vote for a proposal hash that was already committed. In this case, it will receive the commit message from the peer that received the vote. This process is called \textit{commit forwarding}.

%--------------------- | Proof | ---------------------%

\section{Proof of Byzantine Fault Tolerance}
\label{sec:proof}

In order to prove the BFT property of the consensus algorithm, we must prove the safety and liveness of the algorithm.
The safety property is separated in two parts. First, it should be impossible to reach consensus for different hashes for the same proposal on honest peers. Second, there should be time bounds (in terms of rounds) for honest peers in the network. The liveness property provides a guarantee for a valid client's transaction to appear in some consensus round. Proof of liveness is done by construction.

%The following assumption about the Ordering Service will follow for the algorithm steps.

\begin{assumption}
\label{assumption:rejectCase}
    The proof doesn't cover the reject case. If a reject case is detected at any peer it means that the BFT environment is broken, so this situation violates the BFT assumption.
\end{assumption}

%--------------------- | Safety | ---------------------%

\begin{lemma}\label{lemma:safety}
    Safety: the output of consensus for all honest peers will eventually be the same.
\end{lemma}

\begin{proof}
    We have to proof that all honest peers in the network eventually will be in the same state.
    We can prove this by separating into two parts:
    \begin{itemize}
    \label{safety:safety_statements}
        \item[a)]~\label{safety:imposible} Impossible to achieve different states $S_1$ and $S_2$ for round $r$ for honest peers.
        \item[b)]~\label{safety:lag} If the maximum round in the network is $r + 1$, then an honest peer may only be in rounds $r$ and $r + 1$.
    \end{itemize}
    Prove a) by contradiction.
    Assume that peers $P_1$ and $P_2$ are honest and they commit different hashes $H_1$ and $H_2$ respectively.
    The commit message of $P_1$ contains at least $2f+1$ votes by commit message definition~\ref{subsec:commit}.
    The same situation applies for peer $P_2$: the commit message contains at least $2f+1$ votes for $H_2$.
    But this situation is impossible because the hash and cryptography note~\ref{assumption:crypto} and our assumptions of the BFT environment: only $f$ out of $3f+1$ peers may generate non-unique votes for a different hash, but we have at least $2f +1$ peers are not honest.

    Proving b) by the definition.
    First, we need to understand how different peers appear in different rounds.
    When peer $P_i$ multicasts a commit message, it changes its own state: $P_i: r \rightarrow r + 1$.
    And peer $P_i$ propagates a commit message to $P_j$, where $i \neq j$.
    $P_j$ will receive the commit by asynchronous assumption, but we do not know when.
    So, $P_i$ and $P_j$ will stay in different rounds: $r + 1$ and $r$.
    Now we have to show if peer $P$ that is in round less then $r$, where someone is in round $r + 1$, is faulty peer.
    If someone is in round $r + 1$, that means by definition that the block from round $r$ has already committed.
    So, peer $P$ did not vote for the block in round $r$ because the commit message with a supermajority of votes has committed.
    Thus, in round $r + 1$ it does not have a valid state for voting.
\end{proof}

%--------------------- | Liveness | ---------------------%

\begin{lemma}\label{lemma:liveness}
    Liveness: all correct transactions will be eventually committed.
\end{lemma}

\begin{proof}\label{proof:liveness}
    By the assumption of the Ordering Service (Assumption~\ref{assumption:os}), we guarantee that all valid clients' transactions will appear in some proposal and proposals will be received by all honest peers in the network.

    Assume peer $P_{source}$ achieves a supermajority of votes and multicasts a commit message to the network.
    It follows that two situations that may happen on multicast: a commit will be received by peer $P_{target}$, or not.

    Commit has been received by $P_{target}$:
    liveness is achieved for the target peer because of the asynchronous environment.

    Commit has not been received by $P_{target}$ from $P_{source}$.
    This situation happens when $P_{source}$ is faulty and shut down in the process of committing, or is malicious and does not send it.
    Peer $P_{target}$ is conducting the vote step while it has not received the commit message:
    vote from $P_{target}$ is received by honest peer $P_{honest}$.
    These are the possible situations:
    
    \begin{enumerate}
    \label{enumerate:liveness}
        \item[1.] $P_{honest}$ has not committed anything. In this case $P_{target}$ will repeat the vote step. Eventually some peer will achieve a supermajority of votes because of asynchronous and BFT environment assumptions.
        \item[2.] $P_{honest}$ has reached a number of votes for multicasting a commit message. This is the same case as with multicast of $P_{source}$. This case is the same as the initial assumption with $P_{source}$ peer, where $P_{source} \coloneqq P_{honest}$. Definition of honest peers also guarantee that the commit will be delivered to $P_{target}$.
        \item[3.] $P_{honest}$ has already committed. It directly propagates a commit message to $P_{target}$. However, for some reason, $P_{target}$ did not receive the commit message on multicast. This situation may happen because a supermajority subset already voted for the block and $P_{target}$ was not available due to interruptions in the network connection.
    \end{enumerate}

    Each case of YAC leads to transactions of clients appearing in a proposal and honest peers making an agreement on them.
% Does this handle the case where clients try to double spend?
\end{proof}

\section{Experiments}

This section describes experiments conducted over the Yet Another Consensus implementation in Hyperledger Iroha, an open-source, distributed ledger technology platform written in C++. The codebase used was the latest development branch, after the release of version 1.0 alpha. The system configuration used for the testing is described in the system overview section below, and details of each experiment are enclosed in the corresponding sections.

\subsection{System overview}

For each experiment, 4 virtual private servers (VPS) were used with similar configuration and different geographic locations. Servers were provisioned so that two machines were located in Tokyo, Japan, a single server was in Singapore, and another one was in Los Angeles, USA. Each Hyperleder Iroha peer was deployed in a docker environment, having a docker image created from the development branch at aspecified commit\footnote{\texttt{1adfd0}}, and a PostgreSQL 9.5 docker container for each Iroha docker container. Logs were collected from each container, using Elasticsearch, Logstash, and Kibana stack.

The VPS service provider used was Vultr~\cite{vultr}. Each docker environment at the VPS servers had the following configuration:

\begin{itemize}
\item Kernel Version: 4.4.0-109-generic
\item Operating System: Ubuntu 16.04.3 LTS
\item Architecture: x86\_64
\item CPUs: 4
\item Total Memory: 7.796GiB
\end{itemize}

Ansible Playbook was used to deploy an arbitrary number of Iroha peers, varying from 4 peers, up to 64.

\subsection{Vote step delay}

The purpose of this experiment was to deduce an optimal value for the vote step delay, which is the period of time for each peer to decide on the proposal to become the next block in the chain by voting. We tested several values in different network configurations to discover the relationship between the number of peers in the network, vote step delay parameter, and the number of proposals processed per second (system throughput). Figure~\ref{fig:vote_step} shows the median value of system throughput based on ten trials.

\begin{figure}[htbp]

\centering
\resizebox{\columnwidth}{!}{\includegraphics{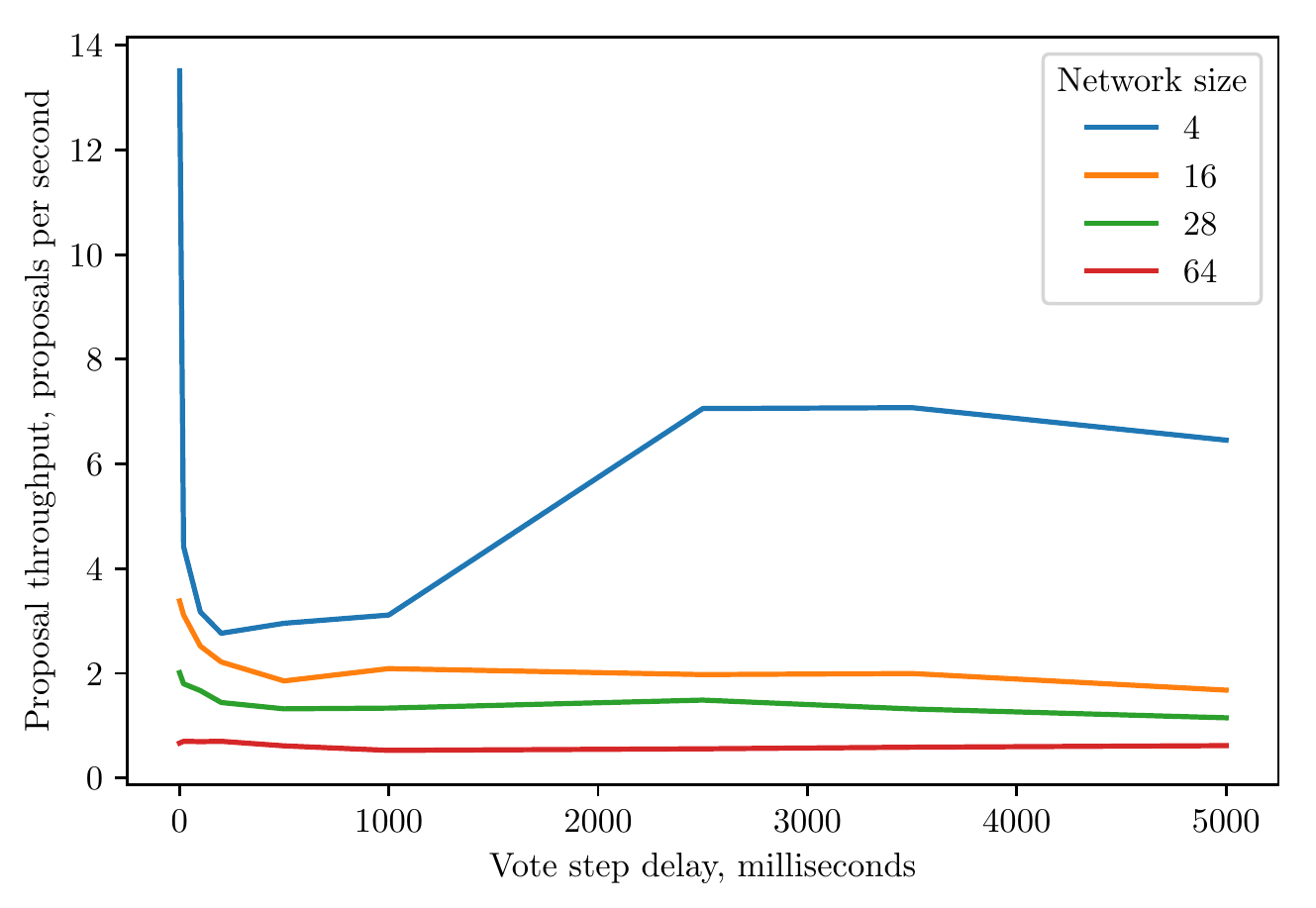}}
\caption{Proposal throughput for varying vote step delay}
\label{fig:vote_step}
\end{figure}

The experiment shows that in a small network size configuration (4 peers), small vote delay values (1 millisecond) are effective to increase the throughput. However, for bigger values of network size (from 16 to 64 peers), small vote delays impose limitations on the ability of peers to reach consensus on a proposal and thus are not improving the throughput of the system. The number of peers that exhibited unstable behavior across the 10 trials in the experiments are displayed in Table~\ref{tab:vote_step}.

\begin{table}[htbp]
\centering
\caption{Number of faulty nodes in conducted experiments}
\begin{tabular}{l|llll}
\multicolumn{1}{c}{Vote step delay,} & \multicolumn{4}{c}{Network size} \\\hline
\multicolumn{1}{c}{milliseconds}     & 4     & 16     & 28     & 64     \\\hline
1                                    & 0     & 1      & 3      & 15      \\
20                                   & 0     & 0      & 0      & 10     \\
100                                  & 0     & 0      & 0      & 4      \\
200                                  & 0     & 0      & 0      & 3      \\
500                                  & 0     & 0      & 0      & 0      \\
1000                                 & 0     & 0      & 1      & 2      \\
2500                                 & 0     & 0      & 0      & 0      \\
3500                                 & 0     & 1      & 0      & 0      \\
5000                                & 0     & 0      & 1      & 1     
\\
\end{tabular} 

\label{tab:vote_step}
\end{table}
\section{Discussion}

To propagate data between validating peers in the network takes time, thus a vote step delay that is too short can cause peers to time out and fail consensus. Our results show that a general trend is that the smallest value (1ms) of vote delay in the system shows the best results for throughput (number of proposals per second), however, the chance of consensus being reached in the network is less compared to larger vote step delay values. Ideally, each network configuration should be tested before the peers agree and use the same value for vote step delays in order to find a peak where the network of peers is more likely to reach consensus and yet the throughput does not start to degrade. In our experiments such extreme points were:

\begin{itemize}
\item 4 nodes  -- 1 millisecond
\item 16 nodes -- 1 millisecond
\item 28 nodes -- 1 millisecond
\item 64 nodes -- 20 milliseconds
\end{itemize}

Thus it is important to tailor the vote step delay	to match the number of validating peers in the network.

%\subsection{Transaction throughput}

%\subsection{Breakdown of elapsed time - collaboration, ordering, validation}

%TODO:  10/03/2018 add words about implementation and Iroha %

%TODO:  10/03/2018 words about extensions of yac: implementation OS, pop/yale, optimization of vote propagation %
%TODO:  10/03/2018 words about application in permission-less: domains, proof of stake %
%In small network like 4 peers vote step delays are causing noticeable impact on the throughput as the network size is small and consensus takes less time. In bigger networks the impact is not that noticeable and moreover causes consensus to fail if 1ms is set for example to 64 or 28 nodes, check failure stats

\section{Conclusion}

We presented YAC, a novel Byzantine fault tolerant consensus algorithm for blockchain systems. Using voting on block proposals, YAC is able to guarantee safety and liveness for transaction processing, given that not more than $f$ faulty validating peers are present out of at least $3f+1$ peers on the network. This paper provides an overall description of the algorithm execution in different cases. Also, we followed the steps of a YAC execution example where peers have different states. Empirical results using the open source implementation of YAC in Hyperledger Iroha show that the algorithm can scale to dozens of validating peers, however the delay at the vote step must be adjusted for the number of validating peers in order to reduce exhibited faults by peers. 

%Future work should explore incorporating stake-weighted voting into YAC in order to provide an opportunity to increase liveness in the face of greater than $f$ faults.

\ifCLASSOPTIONcompsoc
\else
\fi

%The authors would like to thank the National Bank of Cambodia for collaborative development and research on the technology presented in this paper.

\ifCLASSOPTIONcaptionsoff
  \newpage
\fi



\bibliographystyle{IEEEtran}
\bibliography{yac-paper}
%\vfill


\end{document}